\documentclass[12pt]{article}
\usepackage[pdftex,bookmarksopen=true,bookmarks=true,unicode,setpagesize]{hyperref}

\hypersetup{colorlinks=true,linkcolor=black,citecolor=black}

\usepackage{amssymb, amsmath, amsthm}

\numberwithin{equation}{section}

\allowdisplaybreaks
\usepackage{cite}

\renewcommand\i{{\rm 1\kern -.3600em 1}}

\newtheorem{theorem}{Theorem}[section]
\newtheorem{corollary}[theorem]{Corollary}
\newtheorem{lemma}[theorem]{Lemma}
\newtheorem{proposition}[theorem]{Proposition}

\theoremstyle{remark}

\theoremstyle{remark}

\theoremstyle{remark}
\newtheorem{remark}[theorem]{Remark}

\newcommand\M{\mathcal{M}}
\newcommand\R{\mathbb{R}}
\newcommand\Ga{\Gamma}
\newcommand\ga{\gamma}

\newcommand{\om}{\omega}
\newcommand\X{{\R^d}}
\newcommand\D{\mathcal{D}}

\begin{document}

\vspace{-20mm}
\begin{center}{\Large \bf
Finite difference calculus in the continuum}
\end{center}

{\large Dmitri Finkelshtein}\\ Department of Mathematics, Swansea University, Bay Campus,  Swansea SA1 8EN, U.K.;
e-mail: \texttt{d.l.finkelshtein@swansea.ac.uk}\vspace{2mm}

{\large \fbox{Yuri Kondratiev}}\vspace{2mm}

{\large Eugene Lytvynov}\\ Department of Mathematics, Swansea University, Bay Campus,  Swansea SA1 8EN, U.K.; e-mail: \texttt{e.lytvynov@swansea.ac.uk}\vspace{2mm}

{\large Maria Jo\~{a}o Oliveira}\\ DCeT, Universidade Aberta, 1269-001 Lisbon, Portugal; CMAFcIO, University of Lisbon, 1749-016 Lisbon, Portugal; e-mail: \texttt{mjoliveira@ciencias.ulisboa.pt}\vspace{2mm}

{\small
\begin{center}
{\bf Abstract}
\end{center}
\noindent  The paper describes known and new results about finite difference calculus on configuration spaces. We describe finite difference geometry on configuration spaces, connect finite difference operators with cannonical commutation relations, find explicit form for certain finite difference Markov generators on configuration spaces, and describe spaces of Newton series defined over the configuration spaces. 
\vspace{2mm}

\vspace{2mm}


}

\section{Introduction}

This article was written with the enthusiasm, dedication, and deep scientific insight of our collaborator, co-author, and dear friend, Yuri Kondratiev, whose untimely passing has been a great loss to us and to the entire mathematical community.

\bigskip

The set $\mathbb{N}=\{0,1,2,\ldots\}$ of natural numbers is the cornerstone of number theory, combinatorics, discrete probability, and many other areas of mathematics. The obvious inclusion $\mathbb{N} \subset \mathbb{R}$ allows for the extension of some combinatorial quantities on $\mathbb{N}$ to polynomials on $\mathbb{R}$. For example, for each $n \in \mathbb{N}$, the binomial coefficients $\binom{n}{k} = \frac{n!}{k!(n-k)!}$ and the Pochhammer symbol $(n)_k = n(n-1)\dotsm(n-k+1)$, defined for $0 \leq k \leq n$, can be extended to polynomials of order $n$ of $t \in \mathbb{R}$:
\[ 
\binom{t}{k} := \frac{(t)_k}{k!} := \frac{t(t-1)\dotsm (t-k+1)}{k!}, \quad t \in \mathbb{R}.
\]
(For $k = 0$, $\binom{t}{0} = (t)_0 = 1$.) The polynomials $(t)_k$ are called \emph{falling factorials}, and they exhibit many properties similar to their combinatorial counterparts (see, e.g., \cite{Roman}).

Differential calculus on $\mathbb{R}$ is determined by the differential operator $\partial : f \to f'$, which acts on monomials as follows: $\partial t^n = nt^{n-1}$. The corresponding role for the falling factorials is taken by the finite difference operator
\[ 
(D^+f)(t) := f(t+1) - f(t)
\]
defined on functions $f : \mathbb{R} \to \mathbb{R}$. Indeed, we then have
\[ 
D^+ (t)_n = n (t)_{n-1}.
\]
One can also consider the ``dual'' operator (with respect to formal integration over $\mathbb{R}$)
\[ 
(D^- f)(t) := f(t-1) - f(t),
\]
where $D^- (t)_n = -n (t-1)_{n-1}$.

Equivalently, the falling factorials can be defined by their (exponential) generating function
\[ 
e_\lambda(t) := \sum_{n=0}^\infty \frac{\lambda^n}{n!} (t)_n = (1+\lambda)^t = e^{t \log(1+\lambda)}
\]
which satisfies the finite difference Malthus equation:
\[ 
D^+ e_\lambda(t) = \lambda e_\lambda(t).
\]
In this way, we arrive at the framework of finite difference calculus (see, e.g., \cite{FS, Gelfond, Milne-Thomson}).

The properties of monomials and falling factorials are deeply related to the properties of the corresponding differential and difference operators, respectively. The generalization of these properties to other polynomial sequences and the corresponding, so-called, lowering operators led to the development of \emph{umbral calculus} (see, e.g., \cite{Roman} and the references in \cite{Umbral}).

Another important interpretation of the set $\mathbb{N}$ of natural numbers is the state space of a system in, e.g., population ecology, where $n \in \mathbb{N}$ represents the size of the population. The random changes of such a system are described by distributions of discrete random variables, Markov chains, birth-and-death processes, etc. (see, e.g., \cite[Chapter~II]{Feller}). This representation means that one ignores the distribution of the members of a population in their location space, e.g. $\X$ (to slightly simplify the considerations below). If one does want to include the information about the location of individuals, then it is natural to replace $\mathbb{N}$ with the space of (locally finite) configurations in $\X$, denoted by $\Gamma(\X)$. Hence, one aims at a generalization of certain key structures of combinatorics and analysis on $\mathbb{N}$ to those on $\Gamma(\X)$. For an attempt to achieve some of these aims, we refer to our recent papers \cite{Stirling, Umbral}.

The present paper aims at reviewing and developing some basics of finite difference calculus on the configuration space. The paper is organsied as follows. In Section \ref{vyrts6ueie}, we define the configuration space $\Gamma(\X)$, the falling factorials on $\Gamma(\X)$, and recall some well-known facts related to Poisson measures on $\Gamma(\X)$. In Section \ref{drtsy5a3w}, we define two kinds of finite difference operators acting on functions on $\Gamma(\X)$. We also introduce elements of differential geometry (or rather finite difference geometry) on $\Gamma(\X)$ that arise through these finite difference operators. In Section \ref{cxdtsqdhqiu}, we prove that the two kinds of finite difference operators defined in the preceding section, jointly with the operators of multiplication by a monomial, naturally lead to a representation of the canonical commutation relations. In Section \ref{yutfcr}, we consider several classes of Markov generators on $\Gamma(\X)$ that are built upon the finite difference operators. Finally, in Section~\ref{vyei67}, we construct two spaces of functions on $\Gamma(\X)$ that are given through their Newton series, i.e., convergent series of falling factorials.

\section{Elements of spatial combinatorics}\label{vyrts6ueie}

Denote by $\Ga(\X)$ the set of all locally finite configurations (subsets) from $\X$, 
$$
\Ga(\X): =\{ \ga\subset \X \mid |\ga\cap K|<\infty  \text{ for any compact }K\subset \X\}.
$$
Here $|\ga\cap K|$ denotes the number of elements of the set $\ga\cap K$.
It is this set that we will use as a substitution for $\mathbb N$ in  spatial (continuous)  combinatorics.  

Configuration spaces possess various  discrete and continuous features. For differential geometry, differential operators, and diffusion processes, see e.g.\ \cite{AKR}.  On the other hand, discreteness of an individual configuration makes it possible to introduce a proper analog of finite difference calculus.

If $\Ga(\X)$ plays the role of $\mathbb N$, then a natural counterpart of $\R$ can be given by the space $ \mathcal{M}(\X)$ of signed Radon measures on $\X$ with the henceforth fixed Borel $\sigma$-algebra on $\X$. The embedding $\Ga(\X)\subset \mathcal{M}(\X)$ is achieved by interpreting configurations as discrete  Radon measures on $\X$,
$$
\Ga(\X) \ni  \ga \mapsto \ga(dx):= 
\sum_{y\in\ga} \delta_y (dx) \in \M(\X).
$$ Here $\delta_y$ is the Dirac measure with mass at $y$.
Even a wider continuous extension is possible if we recall that $\M(\X)$ is naturally embedded into the space $\mathcal D'(\X)$ of distributions which is the dual of the space $\mathcal D(\X)=C_0^\infty(\X)$ of infinitely differentiable functions on $\X$ with compact support endowed with the natural topology. We will denote by $\langle\omega,\xi\rangle$ the dual pairing between $\omega\in\mathcal D'(\X)$ and $\xi\in \mathcal D(\X)$.

We will now introduce an analog of  a generating function from
classical combinatorics. 
For a test function  
$-1< \xi \in  \mathcal{D}(\X)$, consider the function
\begin{equation}\label{vyrsw6u}
E_{\xi}  (\omega) := e^{\langle\omega,\ln (1+\xi)\rangle} ,\quad \omega \in \mathcal{D}'(\X).
\end{equation}
The power decomposition with respect to  $\xi$ gives
\begin{equation}
\label{NP}
E_\xi(\omega)=\sum_{n=0}^\infty  \frac{1}{n!}  \langle (\omega)_n,\xi^{\otimes n}\rangle.
\end{equation}
Generalized kernels $(\omega)_n \in \mathcal{D}'(\R^{d})^{\odot n}$ are called  falling factorials on $\mathcal{D}'(\X)$. Here $\odot$ denotes the symmetric  tensor product, and we use the same notation $\langle \cdot,\cdot\rangle$ for the dual pairing between $\mathcal{D}'(\R^{d})^{\odot n}$ and $\mathcal{D}(\R^{d})^{\odot n}$.
For further details, see \cite{Umbral,Stirling}.

We may also define binomial coefficients on $ \mathcal{D}'(\X)$ by
\begin{equation}\label{gvgvf}
\binom \omega n:= \frac{(\omega)_n}{n!}.
\end{equation}

\begin{theorem}[\!\!{\cite{Umbral}}]
(i) For $\om \in \mathcal D' (\X)$, 
\begin{gather}
(\om)_0= 1,\quad (\om)_1 =\om,\notag\\
(\om)_n(x_1,\dots,x_n)= \om(x_1) (\om(x_2)-\delta_{x_1}(x_2))
\dotsm (\om(x_n) -\delta_{x_1}(x_n) - \dots -  \delta_{x_{n-1}} (x_n)).\label{cxrw435q}
\end{gather}
In the special case $\om=\ga =\sum_{i\in\mathbb N}\delta_{x_i}\in\Gamma(\X)$,
  \begin{equation}\label{vuydts}
 \binom \ga n  =  \sum_{\{i_1.\dots, i_n\} \subset \mathbb N}
 \delta_{x_1} \odot \dots \odot \delta_{x_n}.
 \end{equation}
 In particular, $\binom \gamma n$ is a symmetric (discrete) Radon measure on $(\X)^n$. 

(ii) For any $\omega_1,\omega_2\in\mathcal D'(\X)$, 
\begin{equation}\label{gu7dei6}
(\omega_1 + \omega_2)_n= \sum_{k=0}^n \binom nk (\omega_1)_k \odot (\omega_2)_{n-k}.
\end{equation}
\end{theorem}

Formula \eqref{gu7dei6} states that the falling factorials on $ \mathcal D' (\X)$ have the binomial property. This formula can also be written in the form of the  Chu--Vandermond identity on $ \mathcal D' (\X)$:
$$\binom{\omega_1 + \omega_2}n=\sum_{k=0}^n \binom{\omega_1}k\odot \binom{\omega_2}{n-k}.$$

In our considerations, it will be useful to use the space of polynomials on $\D'(\X)$.
The polynomials are defined as the linear hull of monomials
$$
\langle\omega^{\otimes n},f^{(n)}\rangle,\quad \omega \in \D'(\X),
$$
where $f^{(n)} \in\D(\X)^{\odot n}$.

Probability measures on the configuration space $\Ga(\X)$, equipped with the cylinder $\sigma$-algebra, give states of continuous particle systems. The Poisson measures form arguably the simplest, yet most important class of probability measures on  $\Ga(\X)$.
 Take a diffuse Radon measure $\sigma$ on
$\X$. The Poisson measure $\pi_\sigma$ is defined via its Laplace
transform:
\begin{equation*}
\int_{\Ga(\X)} e^{\langle\ga,\,f\rangle} \pi_\sigma (d\ga)= \exp\bigg(\int_\X (e^{f(x)} -1)\sigma(dx) \bigg)
\end{equation*}
for all $f\in C_0(\X)$ (the set of continuous functions on $\X$ with compact support).  An alternative definition of $\pi_\sigma$ is given by the Mecke identity:
\begin{equation}
\label{Georgii}
\int_{\Ga(\X)} \sum_{x\in\ga} F(x,\ga) \pi_{\sigma}(d\ga) =
\int_{\Ga(\X)} \int_{\X}  F(x,\ga\cup x) \sigma(dx) \pi_\sigma(d\ga)
\end{equation}
for all appropriate functions $F(x,\ga)$, e.g.\ for all measurable non-negative ones. 

\begin{lemma}
\label{FFN}
Let a function $f^{(n)} \in L^1 ((\X)^n, \sigma^{\otimes n})$ be symmetric. Then, for $\pi_\sigma$-a.a.\ $\gamma\in\Ga(\X)$, we have $\big\langle(\gamma)_n,|f^{(n)}|\big\rangle<\infty$. Furthermore, the function $\big\langle(\cdot)_n,f^{(n)}\big\rangle $ belongs to 
 $L^1(\Ga(\X), \pi_\sigma)$ and
\begin{equation*}
\int_{\Ga(\X)} \bigl\lvert \langle(\ga)_n,f^{(n)}\,\rangle\bigr\rvert\, \pi_{\sigma} (d\ga) \leq \int_{(\X)^n} |f^{(n)}(x_1,\dots,x_n)| \sigma^{\otimes n}(dx_1\dotsm dx_n).
\end{equation*}

\end{lemma}

\begin{proof}
Although this result is well known, for the reader's convenience we present its proof. 
It is sufficient to prove that, if additionally $f^{(n)}\ge0$, then  
\begin{equation}
\label{FgftyF}
\int_{\Ga(\X)} \langle(\ga)_n,f^{(n)}\rangle\, \pi_{\sigma} (d\ga) =\int_{(\X)^n} f^{(n)}(x_1,\dots,x_n) \sigma^{\otimes n}(dx_1\dotsm dx_n).
\end{equation}
We note that \eqref{gvgvf} and \eqref{vuydts} imply
\begin{equation}\label{vyrsw53hxf}
 \langle(\ga)_n,f^{(n)}\rangle=\sum_{x_1\in\gamma}\sum_{x_2\in\gamma\setminus \{x_1\}}\dotsm \sum_{x_n\in\gamma\setminus \{x_1,x_2,\dots,x_{n-1}\}}f^{(n)}(x_1,\dots,x_n).
\end{equation}
Integrating the left and right hand sides of equality \eqref{vyrsw53hxf} with respect to $\pi_\sigma $ and applying $n$ times the Mecke identity \eqref{Georgii}, we get \eqref{FgftyF}.
\end{proof}

The following lemma, which follows directly from  \eqref{vyrsw53hxf} will be useful below.

\begin{lemma}\label{crdstj6sw5a}
Let a function $f^{(n+1)}:(\X)^{n+1}\to\R$ have a compact support. Then, for each $\gamma\in\Gamma(\X)$,
$$\sum_{x\in\gamma}\langle(\gamma\setminus x)_n,f^{(n+1)}(x,\cdot)\rangle=\langle(\gamma)_{n+1},f^{(n+1)}\rangle.$$
In the above formula, one can replace the function $f^{(n+1)}$ with its symmetrization.
\end{lemma}

\section{Finite difference geometry in the continuum}\label{drtsy5a3w}

Let $F:\Gamma(\X)\to\mathbb R$ and let $\gamma\in\Gamma(\X)$. For each $x\in \ga$, we define an elementary Markov death operator (death gradient)
$$
D^-_x F(\ga):= F(\ga\setminus x)-F(\ga).
$$
We define the tangent space $T^-_{\ga} (\Ga):= L^2 (\R^d, \ga)$. Then, for $\psi \in C_0(\R^d)$,
\begin{align*}
D^-_{\psi} F(\ga):&=\sum_{x\in\ga} \psi(x) D^-_x F(\ga)=\sum_{x\in\ga} \psi(x) (F(\ga\setminus x)-F(\ga))
\\&=\langle \psi, D^-_\cdot F(\ga)\rangle_{T_\ga^-(\Ga)}
\end{align*}
is a directional (finite difference)  derivative.

Similarly, we define, for each $x\in \R^d$, 
$$
D^+_x F(\ga):= F(\ga \cup x)-F(\ga)
$$
and the tangent space $T^+_{\ga} (\Ga):= L^2(\R^d, dx)$. Then, for $\psi \in C_0(\R^d)$ and an appropriate function $F$, 
\begin{align*}
D^+_\psi F(\ga):&=  \int_{\R^d} \psi(x) D^+_x F(\ga) dx= \int_{\R^d} \psi(x) (F(\ga \cup x)-F(\ga)) dx
\\&= \langle \psi,D_\cdot^+ F(\ga)\rangle_{T_\ga^+(\Ga)}
\end{align*}
is another directional (finite difference) derivative.

The joint tangent space at $\ga\in \Ga(\X)$ is defined by
$$
T_\ga (\Ga) :=  T^+_{\ga} (\Ga) \oplus T^-_{\ga} (\Ga)
$$
and the tangent bundle is then defined as 
$$
T(\Ga):= \bigcup_{\ga\in \Ga(\X)} T_\ga (\Ga).
$$

A vector field $V\in T(\Ga)$ is a section of the tangent bundle, thus 
$$
V(\ga,x)= (V^+ (\ga,x), V^- (\ga,x))
$$
with 
$V^\pm (\gamma,\cdot)\in T_\gamma^\pm (\Ga)$. In particular,
for a   function $F:\Ga(\X) \to \R$, we define $ (D^+ F) (\ga, x)=( D^+_x F)(\ga) $
and $ (D^- F)(\ga,x)= (D^-_x F)(\ga)$, and consider the the vector field $DF=(D^+ F, D^- F)$, which is naturally called the finite difference gradient of the function $F$.

The derivatives along a vector field $V=(V^+,V^-)$  are given as follows: 
$$
(D^+_{V^+}\,F)(\ga) = \int_{\X} V^+(\ga,x) (F(\ga\cup x) - F(\ga))dx=\langle V^+(\ga,\cdot),D_\cdot^+F(\ga)\rangle_{T_\ga^+(\Ga)},
$$
$$
(D^-_{V^-}\,F) (\ga) = \sum_{x\in\ga} V^-(\ga,x) (F(\ga\setminus x) -F(\ga))=\langle V^-(\ga,\cdot),D_\cdot^-F(\ga)\rangle_{T_\ga^-(\Ga)},
$$
$$
D_V F = D^+_{V^+}F + D^-_{V^-} F=\langle V, DF\rangle_{T(\Ga)}.
$$

Similarly to classical differential geometry, we would like to introduce a divergence
for a vector field. Having in mind the usual duality relation between the gradient and the divergence, we will use an integration  over $\Gamma(\X)$ to establish such a duality.  As a reference measure on $\Gamma(\X)$, let us choose the Poisson measure $\pi=\pi_\sigma$ with the Lebesgue density $\sigma(dx)=dx$.

A function $F:\Gamma(\X)\to\R$ is called local if there exists a bounded set $\Lambda\subset\X$ such that $F(\gamma)=F(\gamma\cap\Lambda)$ for all $\gamma\in\Gamma(\X)$. In particular, $(D^+F)(\gamma,x)=(D^-F)(\gamma,x)=0$ for  all $\gamma\in\Gamma(\X)$ and $x\in\Lambda^c$.

For an appropriate vector field $V$, its divergence $\operatorname{Div} V:\Gamma(\X)\to\R$ is defined by the relation
$$
\int_{\Ga(\X)} (D_V F)(\ga)\, \pi(d\ga)= - \int_{\Ga(\X)} (\operatorname{Div} V)(\ga) F(\ga)\, \pi(d\ga),
$$ which holds for all measurable bounded local functions $F:\Gamma(\X)\to\R$. 

The following theorem establishes the precise form of the divergence. 

\begin{theorem} Let $V=(V^+,V^-)$ be a vector field that satisfies the following assumption: for each bounded measurable $\Lambda\subset\X$, 
\begin{equation}\label{utf7i}
V^+(\gamma,x),\ V^-(\gamma\cup x,x)\in L^1(\Gamma(\X)\times\Lambda,d\pi(\gamma)\,dx)\end{equation}
and
\begin{equation}\label{zrweitr}
 V^-(\gamma\cup x,x)-V^+(\gamma,x)\in L^1(\Gamma(\X)\times\X ,d\pi(\gamma)\,dx).\end{equation}
Then
\begin{equation}\label{vcyrstwu56}
(\operatorname{Div} V)(\ga)= \sum_{x\in \ga} \big(V^+ (\ga\setminus x,x)-V^- (\ga,x)\big) +\int_{\X} \big(V^- (\ga \cup x,x)-V^+ (\ga,x)\big) dx.
\end{equation}

\end{theorem}

\begin{proof}
Using the Mecke identity \eqref{Georgii}, we have, for each measurable bounded local function $F$ on $\Gamma(\X)$,
\begin{align}
&\int_{\Ga(\X)} (D_V F)(\ga)\, \pi(d\ga)\notag\\
&=\int_{\Gamma(\X)}\int_{\X}\big(V^+(\gamma,x)(F(\gamma\cup x)-F(\gamma))+V^-(\gamma\cup x,x)(F(\gamma)-F(\gamma\cup x))dx\,\pi(d\gamma)\notag\\
&=\int_{\Gamma(\X)}\int_{\X} \big(V^-(\gamma\cup x,x)-V^+(\gamma,x))F(\gamma)\,dx\,\pi(d\gamma)\notag\\
&\quad+\int_{\Gamma(\X)}\int_{\X}\big(V^+(\gamma,x)-V^-(\gamma\cup x,x)\big)F(\gamma\cup x)\,dx\,\pi(d\gamma)\notag\\
&=\int_{\Gamma(\X)}\int_{\X} \big(V^-(\gamma\cup x,x)-V^+(\gamma,x))\,dx\,F(\gamma)\pi(d\gamma)\notag\\
&\quad+\int_{\Gamma(\X)}\sum_{x\in\gamma}\big(V^+(\gamma\setminus x,x)-V^-(\gamma ,x)\big)F(\gamma) \,\pi(d\gamma).\label{cfystu5}
\end{align}
Note that conditions \eqref{utf7i}, \eqref{zrweitr} justify the above calculations. Formula \eqref{cfystu5} implies the statement.
\end{proof}

Assume that the vector field $V$ satisfies the following symmetry relation:
\begin{equation}\label{cxryw4u6}
V^+(\ga,x) =  - V^- (\ga\cup x,x).
\end{equation}
In particular, relation \eqref{cxryw4u6} holds for the gradient $V=DF$ of a function  $F:\Gamma(\X)\to \R$.
Note that, if \eqref{cxryw4u6} holds, then the conditions \eqref{utf7i}, \eqref{zrweitr} become
$$V^+(\gamma,x)\in L^1(\Gamma(\X)\times\X ,d\pi(\gamma)\,dx).$$
Formulas \eqref{vcyrstwu56} and \eqref{cxryw4u6} imply 
\begin{equation}\label{cyrrde6ie}
(\operatorname{Div} V)(\ga)= 2 \sum_{x\in\ga} V^+(\ga\setminus x) -2\int_{\X} V^+(\ga,x) dx.
\end{equation}

Denote by $D(\Delta)$ the linear space of all measurable functions $F:\Gamma(\X)\to \R$ that satisfy 
$$F(\gamma\cup x)-F(\gamma)\in L^1(\Gamma(\X)\times\X ,d\pi(\gamma)\,dx).$$
Then, we define a linear operator $\Delta$ with domain $D(\Delta)$ by $\Delta F=\operatorname{Div}D F$. The operator $\Delta$ is called the finite difference Laplacian. It follows from \eqref{cyrrde6ie} that 
$$
(\Delta F) (\ga)= - 2\sum_{x\in \ga} \big(F(\ga\setminus x)-F(\ga)\big) -2\int_\X \big(F(\ga\cup x) -F(\ga)\big)dx.
$$

\begin{remark}
In our considerations, we used the Lebesgue measure on $\X$ as the fixed intensity measure of the Poisson measure. Replacing the Lebesgue measure by
an arbitrary diffuse Radon measure $\sigma(dx)$, we can immediately extend the above formulas to the case of the general measure~$\sigma$.
\end{remark}

\section{Finite difference operators and the canonical commutation relations}\label{cxdtsqdhqiu}

By formulas \eqref{vyrsw6u} and \eqref{NP}, we have for each $\xi\in\mathcal D(\X)$, $\xi>-1$, 
$$E_\xi(\gamma)=\sum_{n=0}^\infty \frac1{n!}\langle (\gamma)_n,\xi^{\otimes n}\rangle =e^{\langle\gamma,\ln (1+\xi)\rangle} ,\quad \gamma \in\Gamma(\X).$$
Hence, for $\psi \in C_0(\R^d)$,
$$
D^+_\psi E_\xi(\ga)=\int_{\X} \psi(x)\xi(x)\,dx\, E_\xi (\ga).
$$
This easily implies that
\begin{equation}\label{vfydst}
D^+_\psi \langle (\gamma)_n, f^{(n)}\rangle =n\int_{\X}\psi(x)\big\langle (\gamma)_{n-1}, f^{(n)}(x,\cdot)\big\rangle\,dx.\end{equation}
Here and below we assume that $f^{(n)}$ belongs to $C_0(\X)^{\odot n}$, the space of continuous symmetric functions on $(\X)^n$ with compact support.  Formula \eqref{vfydst} means that $D^+_\psi $ is a lowering operator for the falling factorials, cf.\ \cite[Subsection 6.1]{Umbral}. Formula \eqref{vfydst} also implies that, for each $x\in\X$,
\begin{equation}\label{vcyds6ue4}
D^+_x \langle (\gamma)_n, f^{(n)}\rangle =n \big\langle (\gamma)_{n-1}, f^{(n)}(x,\cdot)\big\rangle.
\end{equation}

\begin{proposition}\label{esawaq4wq}
Let $\psi\in C_0(\X)$ and $f^{(n)}\in C_0(\X)^{\odot n}$. 
Then
\begin{equation}\label{cxxdxdd}
D^-_\psi \langle (\gamma)_n, f^{(n)}\rangle =-\left\langle (\gamma)_n(dx_1\dotsm dx_n),\big(\psi(x_1)+\dots+\psi(x_n)\big) f^{(n)}(x_1,\dots,x_n)\right\rangle.
\end{equation}
\end{proposition}

\begin{proof}Fix $\gamma\in\Gamma(\X)$ and $x\in\X$. Then \eqref{vcyds6ue4} implies
$$D^+_x \langle (\gamma\setminus x)_n, f^{(n)}\rangle =n \big\langle (\gamma\setminus x)_{n-1}, f^{(n)}(x,\cdot)\big\rangle,$$
from where
\begin{equation}\label{fxdsarwaeersw}
D^-_x \langle (\gamma)_n, f^{(n)}\rangle=-n\big\langle (\gamma\setminus x)_{n-1}, f^{(n)}(x,\cdot)\big\rangle,\end{equation}
and so
\begin{equation}\label{xtrea5yw}
D^-_\psi\,\langle (\gamma)_n, f^{(n)}\rangle=-n\sum_{x\in\gamma} \big\langle (\gamma\setminus x)_{n-1},\psi(x) f^{(n)}(x,\cdot)\big\rangle.\end{equation}
By Lemma~\ref{crdstj6sw5a}, the right-hand side of \eqref{xtrea5yw} is equal to the right-hand side of \eqref{cxxdxdd}.
\end{proof}

\begin{corollary}Let $f^{(n)}\in C_0(\X)^{\odot n}$ and  $F(\gamma)=\langle(\gamma)_n, f^{(n)}\rangle$. Then
\begin{equation*}
\sum_{x\in\gamma}(D_x^-F)(\gamma)=-nF(\gamma).
\end{equation*}
\end{corollary}

\begin{proof}This follows immediately from Proposition \ref{esawaq4wq} if we notice that it remains true in the case where the function $\psi(x)=1$ for all $x\in\X$.
\end{proof}

Formulas \eqref{vfydst} and \eqref{cxxdxdd} allow us to find a connection between the difference derivatives and creation, annihilation and neutral operators in the
symmetric Fock space. To this end let us first recall the definition of the symmetric Fock space. For a real separable Hilbert space $\mathcal H$, the Fock space $\mathcal F(\mathcal H)$ is defined as the real Hilbert space
$$\mathcal F(\mathcal H)=\bigoplus_{n=0}^\infty\mathcal F_n(\mathcal H).$$
Here $\mathcal F_0(\mathcal H)=\R$ and for $n\ge1$, $\mathcal F_n(\mathcal H)=\mathcal H^{\odot n}n!$, i.e., $\mathcal F_n(\mathcal H)$ coincides as a set with $\mathcal H^{\odot n}$, the $n$th symmetric tensor power of $\mathcal H$, and 
$$\|f^{(n)}\|_{\mathcal F_n(\mathcal H)}^2=\|f^{(n)}\|^2_{\mathcal H^{\odot n}}\,n!.$$

Consider the Fock space $\mathcal F(L^2(\X,dx))$. Note that $L^2(\X,dx)^{\odot n}$ is the space of all square integrable symmetric functions  $f^{(n)}:(\X)^n\to\R$. Let $\mathcal F_\mathrm{fin}(C_0(\X))$ denote the dense subset of $\mathcal F(L^2(\X,dx))$ that consists of finite sequences 
\begin{equation}
  f=(f^{(0)},f^{(1)},\dots,f^{(N)},0,0,\dots), \label{finseq}
\end{equation}
where $N\in\mathbb N$ and each $f^{(n)}\in C_0(\X)^{\odot n}$.

For each $\psi\in C_0(\X)$, we define a creation operator $A^+(\psi)$, an annihilation operator $A^-(\psi)$ and a neutral operator $A^0(\psi)$ as follows. These operators act in $\mathcal F_\mathrm{fin}(C_0(\X))$ and for each $f^{(n)}\in C_0(\X)^{\odot n}$, we have $A^+(\psi)f^{(n)}\in C_0(\X)^{\odot (n+1)}$, $A^-(\psi)f^{(n)}\in C_0(\X)^{\odot (n-1)}$, $A^0(\psi)f^{(n)}\in C_0(\X)^{\odot n}$. Furthermore,
\begin{align}
 A^+(\psi)f^{(n)}:&=\psi\odot f^{(n)},\label{yudfy7ed}\\
  \big(A^-(\psi)f^{(n)}\big)(x_1,\dots,x_{n-1}):&=n\int_\X \psi(x)f^{(n)}(x,x_1,\dots,x_{n-1})dx,\label{xsesese}\\
  \big(A^0(\psi)f^{(n)}\big)(x_1,\dots,x_{n}):&=(\psi(x_1)+\dots+\psi(x_n))f^{(n)}(x_1,\dots,x_n).\label{vtudy}
 \end{align}
 The operator $A^-(\psi)$ is the restriction to $\mathcal F_\mathrm{fin}(C_0(\X))$ of the adjoint operator of $A^+(\psi)$ in $\mathcal F(L^2(\X,dx))$, whereas the operator $A^0(\psi)$ is symmetric in $\mathcal F(L^2(\X,dx))$. These operators satisfy the canonical commutation relations: for each $\psi,\xi\in C_0(\X)$,
 \begin{gather}
 [A^+(\psi),A^+(\xi)]=[A^-(\psi),A^-(\xi)]=[A^0(\psi),A^0(\xi)]=0,\notag\\
 [A^-(\psi),A^+(\xi)]=\int_\X \psi(x)\xi(x)\,dx,\notag\\
 [A^+(\psi),A^0(\xi)]=-A^+(\psi\xi),\quad  [A^-(\psi),A^0(\xi)]=A^-(\psi\xi).\label{dr6i}
 \end{gather} 
Here, for linear operators $A$ and $B$, $[A,B]=AB-BA$ is the commutator of $A$ and $B$.

A polynomial on $\Gamma(\X)$ is a function $p:\Gamma(X)\to\R$ of the form
$$p(\gamma)=g^{(0)}+\sum_{n=1}^N \langle\gamma^{\otimes n},g^{(n)}\rangle, \quad \gamma\in\Gamma(\X),$$
where $g^{(0)}\in\R$, $N\in\mathbb N$, and each $g^{(n)}\in C_0(\X)^{\odot n}$.  
We denote by $\mathcal P(\Gamma(\X))$ the space of polynomials on $\Gamma(\X)$.

The following lemma can be easily deduced from \cite[Theorem 4.1]{Stirling}.

\begin{lemma} \label{vfdrye64}The space $\mathcal P(\Gamma(\X))$ consists of all functions of the form
\begin{equation}\label{vcytswu6453e6}
p(\gamma)=f^{(0)}+\sum_{n=1}^N \langle(\gamma)_n ,f^{(n)}\rangle, \quad \gamma\in\Gamma(\X),
\end{equation}
where $f^{(0)}\in\R$, $N\in\mathbb N$, and each $f^{(n)}\in C_0(\X)^{\odot n}$. 
\end{lemma}

Lemma~\ref{vfdrye64} allows us to define a bijective map 
$$I:\mathcal P(\Gamma(\X))\to \mathcal F_\mathrm{fin}(C_0(\X))$$
 by defining, for  a polynomial $p$ given by \eqref{vcytswu6453e6},  $$Ip=(f^{(0)},f^{(1)},\dots,f^{(N)},0,0)\in\mathcal F_\mathrm{fin}(C_0(\X)).$$

\begin{remark}
Since the map $I$ is bijective and the set  $\mathcal F_\mathrm{fin}(C_0(\X))$ is dense in the Fock space $\mathcal F(L^2(\X,dx))$, we can therefore realize the latter Fock space as the closure of $\mathcal P(\Gamma(\X))$ in the Hilbertian norm
$$\|p\|=\bigg(|f^{(0)}|^2+\sum_{n=1}^\infty\|f^{(n)}\|^2_{L^2(\X,dx)^{\odot n}}\,n!\bigg)^{1/2}.$$
\end{remark}

\begin{proposition}\label{dtre564}
For each $\psi\in C_0(\X)$, we have
\begin{align}
I^{-1}A^+(\psi)I&=\langle\cdot,\psi\rangle+D^-_\psi,\label{dtew5u}\\
I^{-1}A^-(\psi)I&=D^+_\psi,\label{ew5u}\\
I^{-1}A^0(\psi)I&=-D^-_\psi. \label{cutfC}
\end{align}
\end{proposition}

\begin{proof}
Formula \eqref{ew5u} follows from \eqref{vfydst} and \eqref{xsesese}. Formula \eqref{cutfC} follows from \eqref{cxxdxdd} and \eqref{vtudy}. Finally, formula \eqref{dtew5u} follows from \eqref{cutfC}  and the recurrence formula
\begin{equation}\label{vfdr6r5}
\langle\gamma,\psi\rangle\langle (\gamma)_n,f^{(n)}\rangle=\langle (\gamma)_{n+1},A^+(\psi)f^{(n)}\rangle+\langle (\gamma)_{n},A^0(\psi)f^{(n)}\rangle,\end{equation}
see e.g.\ formula (5.26) in \cite{Umbral}.
\end{proof}

\begin{remark}
Using Proposition \ref{dtre564} and the canonical commutation relations \eqref{dr6i}, one can easily show that the unital algebra generated by the operators $D^+_\psi$, $D^-_\psi$ and $\langle\cdot,\psi\rangle$ is closed under the commutator $[\cdot,\cdot]$ on it.
\end{remark}

\section{Finite difference Markov generators}\label{yutfcr}
In this section we will consider how certain Markov difference generators act on the falling factorials. 

Let $a\in C_0(\X)^{\odot 2}$ and $m\in\mathbb{N}$. For each $x\in\R^d$, we define $\alpha_x(\cdot):=a(x,\cdot)\in C_0(\R^d)$ and
\begin{equation}\label{rate_c}
c(x,\ga):=\langle (\ga)_m,\alpha_x^{\otimes m}\rangle. 
\end{equation}
For a local function $F:\Gamma(\X)\to\R$, we define the so-called death and birth generators 
\begin{align}
  (L^-F)(\gamma)&=\sum_{x\in\gamma}c(x,\ga\setminus x)(D^-_xF)(\gamma),\label{Lminus}\\
  (L^+F)(\gamma)&=\int_{\R^d}\,c(x,\ga)\,(D^+_xF)(\gamma)\,dx. \label{Lplus}
\end{align}
In other words,
\[
  L^\pm F = D^\pm_{V^\pm}F,
\]
where $V^+(\ga,x)=c(x,\ga)$ and $V^-(\ga,x)=c(x,\ga\setminus x)$.

We start with the following generalization of formula \eqref{vfdr6r5}.
\begin{lemma}
  Let $f^{(n)}\in C_0(\X)^{\odot n}$ and $\psi\in C_0(\R^d)$. We define, for each $1\leq i\leq n$, the operator $B_i(\psi)$ acting in $\mathcal F_\mathrm{fin}(C_0(\X))$ and given by
  \begin{equation}\label{Binew}
    \bigl(B_i(\psi) f^{(n)}\bigr)(x_1,\ldots,x_n)
    :=\biggl(\sum_{\substack{I\subset\{1,\ldots,n\}\\|I|=i}}\psi^{\otimes i}(x_I)\biggr)f^{(n)}(x_1,\ldots,x_n),
  \end{equation}
  where, for any $I\subset \mathbb{N}$, $x_I:=(x_i)_{i\in I}$. (In particular, $B_1(\psi) =A^0(\psi) $.) We set also  $B_0(\psi) f^{(n)}:=f^{(n)}$. Then 
\begin{equation}\label{prodasff}
  \langle (\ga)_n,f^{(n)}\rangle\langle (\ga)_m,\psi^{\otimes m}\rangle
  = \sum_{k=(m-n)_+}^m  \frac{m!}{k!} \bigl\langle (\ga)_{n+k},(B_{m-k}(\psi)f^{(n)})\odot \psi^{\otimes k} \bigr\rangle,
\end{equation}
where $s_+:=\max\{s,0\}$, $s\in\R$.
\end{lemma}
\begin{proof}
  
By \cite{InfKuna}, see also \cite{Stirling}, for any $f\in \mathcal F_\mathrm{fin}(C_0(\X))$ of the form \eqref{finseq}, one can define 
\begin{equation}\label{Ktr}
  (Kf)(\gamma) =\sum_{n}\frac{1}{n!}\langle(\gamma)_n,f^{(n)}\rangle.
\end{equation}
Then, for $f,g\in \mathcal F_\mathrm{fin}(C_0(\X))$, we have $Kf\cdot Kg = K(f\star g)$,
where, for $j\geq1$ and $J=\{1,\ldots,j\}$,
\begin{equation}\label{starconv}
  (f\star g)^{(j)}(x_J):=\sum_{J_1\sqcup J_2\sqcup J_3 = J}f^{(|J_1|+|J_2|)}(x_{J_1\cup J_2})g^{(|J_2|+|J_3|)}(x_{J_2\cup J_3}).
\end{equation}
Then, as usual, identifying $f^{(n)}$ with $(0,\ldots,0, f^{(n)},0,\ldots)\in \mathcal F_\mathrm{fin}(C_0(\X))$, we get
$\langle (\ga)_n,f^{(n)}\rangle  =n! \, (Kf^{(n)})(\ga)$, and hence,
\[
  \langle (\ga)_n,f^{(n)}\rangle \langle (\ga)_m,\psi^{\otimes m}\rangle=
  n!m!  (Kf^{(n)})(\gamma)(K\psi^{\otimes m})(\gamma)= n!m! K(f^{(n)}\star \psi^{\otimes m} )(\gamma).
\]
By \eqref{starconv}, the components of $ f^{(n)}\star\psi^{\otimes m}\in\mathcal F_\mathrm{fin}(C_0(\X))$ are non-zero if and only if $|J_1|+|J_2|=n$ and $|J_2|+|J_3|=m$. Denote $k:=|J_3|$, then $|J|=n+k$. Note that $k\leq m $ and $n\geq |J_2|=m-k$; hence, $k\geq m-n$, and therefore, $(m-n)_+\leq k\leq m$. Then, from \eqref{starconv},
\begin{equation*}
  f^{(n)}\star \psi^{\otimes m}=\sum_{k=(m-n)_+}^m \binom{n+k}{k} (B_{m-k}(\psi)f^{(n)})\odot \psi^{\otimes k},
\end{equation*}
and therefore, by \eqref{Ktr},
\begin{align*}
  &\quad 
  \langle (\ga)_n,f^{(n)}\rangle\langle (\ga)_m,\psi^{\otimes m}\rangle
  \\&= n!m!\sum_{k=(m-n)_+}^m  \frac{1}{(n+k)!}  \binom{n+k}{k} \bigl\langle (\ga)_{n+k},(B_{m-k}(\psi)f^{(n)})\odot \psi^{\otimes k} \bigr\rangle,
\end{align*}
that yields \eqref{prodasff}.
\end{proof}
\begin{theorem}\label{thm:sasaf}
  Let $F(\gamma)=\langle (\gamma)_n,f^{(n)}\rangle$ with $f^{(n)}\in C_0(\X)^{\odot n}$, and let $L^\pm$ be given by \eqref{rate_c}--\eqref{Lplus}. Then
  \[
    (L^-F)(\gamma)= \sum_{k=(m-n+1)_+}^m  \bigl\langle (\ga)_{n+k},h_{m-k}^{(n+k)}\bigr\rangle,
  \]
  where, for $A:=\{x_1,\ldots,x_{n+k}\}$,
  \begin{align*}
    h_{m-k}^{(n+k)}(x_A):=-\frac{m!n!}{(n+k)!}\sum_{j=1}^{n+k}\sum_{\substack{I\subset A\setminus\{j\}\\|I|=k}}\sum_{\substack{J\subset (A\setminus I)\setminus\{j\}\\|J|=m-k}}\prod_{i\in I\cup J}a(x_j,x_i)f^{(n)}(x_{A\setminus I}),
  \end{align*}
and
\[
  (L^+F)(\gamma)= \sum_{k=(m-n+1)_+}^m  \Bigl\langle (\ga)_{n-1+k},\int_{\R^d}p_{m-k}^{(n-1+k)}(x,\cdot)dx \Bigr\rangle,
\]
where, for $A':=\{1,\ldots,n-1+k\}$,
\begin{align*}
  p_{m-k}^{(n-1+k)}(x,x_{A'})  =  \frac{m!n!}{(n-1+k)!} \sum_{\substack{I\subset A'\\|I|=k}}\sum_{\substack{J\subset A'\setminus I\\|J|=m-k}}\prod_{i\in I\cup J}a(x,x_i) f^{(n)}(x,x_{A'\setminus I}).
\end{align*}
\end{theorem}

\begin{remark}
  Another important class of the rates $c(x,\gamma)$ in \eqref{Lminus}--\eqref{Lplus} is given by 
  \[
  c(x,\ga)=f(\langle\ga,a(x,\cdot)\rangle)
  \]
  with $a\in C_0(\R^d)^{\odot 2}$. If $f(x)=\sum_n c_n x^n$ is an analytic function, then $c(x,\ga)=\sum_n c_n \langle\ga,a(x,\cdot)\rangle^{n}$ converges pointwise. Next, each
  \[
  \langle\ga,a(x,\cdot)\rangle^{n}=\sum_{j=1}^{n}\langle
    (\ga)_{n}, \mathbf{S}(n,j) a(x,\cdot)^{\otimes j}\rangle,
  \]
  where $\mathbf{S}(n,j)$ is \emph{the Stirling operators of the second kind}, see \cite{Stirling}. 
Therefore, one can rewrite $L^\pm F$ in this case as well.
\end{remark}

\begin{proof}
By \eqref{fxdsarwaeersw} and \eqref{prodasff},
\begin{align*}
  c(x,\ga\setminus x)(D^-_xF)(\gamma) &=-n\langle (\ga\setminus x)_m,\alpha_x^{\otimes m}\rangle \bigl\langle (\gamma\setminus x)_{n-1}, f^{(n)}(x,\cdot)\bigr\rangle
  \\& =-n \sum_{k=(m-n+1)_+}^m  \frac{m!}{k!} \bigl\langle (\ga\setminus x)_{n-1+k},(B_{m-k}(\alpha_x)f^{(n)}(x,\cdot))\odot \alpha_x^{\otimes k}\bigr\rangle.
\end{align*}
Set $g_{x,m-k}^{(n-1)}(\cdot) :=B_{m-k}(\alpha_x)f^{(n)}(x,\cdot)\in C_0(\R^d)^{\odot (n-1)}$. Then, for $A':=\{1,\ldots,n-1+k\}$,
\begin{align}
  &\quad \bigl(g_{x,m-k}^{(n-1)}\odot \alpha_x^{\otimes k}\bigr)(x_{A'})\notag\\&=\frac{1}{(n-1+k)!}\sum_{\sigma\in S_{n-1+k}}g_{x,m-k}^{(n-1)}(x_{\sigma(1)},\ldots,x_{\sigma(n-1)})\alpha_x^{\otimes k}(x_{\sigma(n)},\ldots,x_{\sigma(n-1+k)})\notag\\
  &=\frac{(n-1)!k!}{(n-1+k)!}\sum_{\substack{I\subset A'\\|I|=k}}g_{x,m-k}^{(n-1)}(x_{A'\setminus I})\prod_{i\in I}a(x,x_i)\label{saqwenlk}
\end{align}
is a function symmetric in $n-1+k$ variables $x_{A'}$ and dependend on $x$. Denoting $x$ by $x_{n+k}$ and considering the symmetrization of the latter function in all variables $x_A$, where $A:=A'\cup\{n+k\}=\{1,\ldots,n+k\}$, by the same arguments, we will get the function
\begin{align*}
  &\quad \frac{(n-1+k)!}{(n+k)!}\sum_{j=1}^{n+k}\bigl(g_{x_j,m-k}^{(n-1)}\odot \alpha_{x_j}^{\otimes k}\bigr)(x_{A\setminus\{j\}})
  \\&=\frac{(n-1)!k!}{(n+k)!}\sum_{j=1}^{n+k}\sum_{\substack{I\subset A\setminus\{j\}\\|I|=k}}g_{x_j,m-k}^{(n-1)}(x_{(A\setminus\{j\})\setminus I})\prod_{i\in I}a(x_j,x_i).
\end{align*}

Then, by Lemma~\ref{crdstj6sw5a}, 
\begin{align*}
  (L^-F)(\gamma)&=-n \sum_{k=(m-n+1)_+}^m  \frac{m!}{k!} \sum_{x\in\ga}\bigl\langle (\ga\setminus x)_{n-1+k},(B_{m-k}(\alpha_x)f^{(n)}(x,\cdot))\odot \alpha_x^{\otimes k}\bigr\rangle\\
  &= \sum_{k=(m-n+1)_+}^m  \bigl\langle (\ga)_{n+k},h_{m-k}^{(n+k)}\bigr\rangle,
\end{align*}
where
\begin{align*}
  h_{m-k}^{(n+k)}(x_1,\ldots,x_{n+k}):=-\frac{m!n!}{(n+k)!}\sum_{j=1}^{n+k}\sum_{\substack{I\subset A\setminus\{j\}\\|I|=k}}g_{x_j,m-k}^{(n-1)}(x_{(A\setminus\{j\})\setminus I})\prod_{i\in I}a(x_j,x_i),
\end{align*}
and since
\begin{equation}\label{e3rqw253}
  g_{x_j,m-k}^{(n-1)}(x_{(A\setminus\{j\})\setminus I})  =\sum_{\substack{J\subset (A\setminus\{j\})\setminus I\\|J|=m-k}}\alpha_{x_j}^{\otimes (m-k)}(x_J)f^{(n)}(x_{A\setminus I}),
\end{equation}
we get the statement.

Similarly, by \eqref{vcyds6ue4}, \eqref{prodasff}, and \eqref{saqwenlk},
\begin{align*}
  c(x,\ga)(D^+_xF)(\gamma) &=n\langle (\ga)_m,\alpha_x^{\otimes m}\rangle \bigl\langle (\gamma)_{n-1}, f^{(n)}(x,\cdot)\bigr\rangle
  \\& =n \sum_{k=(m-n+1)_+}^m  \frac{m!}{k!} \bigl\langle (\ga)_{n-1+k},(B_{m-k}(\alpha_x)f^{(n)}(x,\cdot))\odot \alpha_x^{\otimes k}\bigr\rangle
  \\& = \sum_{k=(m-n+1)_+}^m  \frac{n!m!}{(n-1+k)!} \bigl\langle (\ga)_{n-1+k}, \sum_{\substack{I\subset A'\\|I|=k}}g_{x,m-k}^{(n-1)}(x_{A'\setminus I})\prod_{i\in I}a(x,x_i)\bigr\rangle,
\end{align*}
and applying \eqref{e3rqw253} with $A\setminus\{j\}$ replaced by $A'$, we get the statement.
\end{proof}
\begin{corollary}
Let $m=1$, i.e. $c(x,\gamma)=\sum\limits_{y\in\gamma}a(x,y)$. Then, in the conditions and notations of Theorem~\ref{thm:sasaf}, we have 
  \[
    (L^-F)(\gamma)= \bigl\langle (\ga)_{n},h_{1}^{(n)}\bigr\rangle+\bigl\langle (\ga)_{n+1},h_{0}^{(n+1)}\bigr\rangle,
  \]
where
\begin{align*}
  h_1^{(n)}(x_1,\ldots,x_{n})&= -\sum_{j=1}^{n}\sum_{i\neq j}a(x_j,x_i)f^{(n)}(x_1,\ldots,x_n),\\
  h_0^{(n+1)}(x_1,\ldots,x_{n+1})&=-
  \frac{1}{n+1}\sum_{j=1}^{n+1}\sum_{i\neq j }a(x_j,x_i)f^{(n)}(x_{\{1,\ldots,n+1\}\setminus \{i\}});
\end{align*}
and, for $n\geq2$,
\[
  (L^+F)(\gamma)=  \Bigl\langle (\ga)_{n-1},\int_{\R^d}p_{1}^{(n-1)}(x,\cdot)dx \Bigr\rangle+
  \Bigl\langle (\ga)_{n},\int_{\R^d}p_{0}^{(n)}(x,\cdot)dx \Bigr\rangle,
\]
where
\begin{align*}
  p_1^{(n-1)}(x,x_1,\ldots,x_{n-1})&=n\sum_{j=1}^{n-1}a(x,x_j)f^{(n)}(x,x_{\{1,\ldots, n-1\}}),\\
  p_{0}^{(n)}(x,x_1,\ldots,x_n)&= \sum_{i=1}^n a(x,x_i)f^{(n)}(x,x_{\{1,\ldots, n\}\setminus\{i\}}).
\end{align*}

\end{corollary}

We can also consider the Markov generator of a jump dynamics:
\[
(LF)(\gamma):=\int_{\R^d}\sum_{x\in\ga} \tilde{c}(x,y,\ga\setminus x)\big(F((\gamma\setminus x)\cup y)-F(\gamma)\big)dy,
\]
where $\tilde{c}(x,y,\ga):=\langle (\ga)_m,\beta_{x,y}^{\otimes m}\rangle$ and $\beta_{x,y}(z):=b(x,y,z)$ for some $b\in C_0(\R^d)^{\odot 3}$.
Since
\[
  F((\gamma\setminus x)\cup y)-F(\gamma)=(D^+_yF)(\gamma\setminus x)+(D_x^-F)(\gamma),
\]
this case can be done similarly to the previous one. 

For example, for $\tilde{c}(x,y,\gamma)=a(x,y)$, we will have, for $F(\gamma)=\langle (\gamma)_n,f^{(n)}\rangle$, 
\[
  F((\gamma\setminus x)\cup y)-F(\gamma)=n\langle (\gamma\setminus x)_{n-1},f^{(n)}(y,\cdot)-f^{(n)}(x,\cdot)\rangle,
\] 
and hence, $(LF)(\gamma)=\langle(\gamma)_{n},g^{(n)}\rangle$,
where
$$
g^{(n)}(x_1,\dots,x_n)=\sum_{i=1}^n \int_\X a(x_i-y)\big(f^{(n)}(y,x_{\{1,\ldots, n\}\setminus\{i\}})-f^{(n)}(x_1,\dots,x_n)\big)\,dy.
$$

\section{Spaces of Newton series}\label{vyei67}

Recall that each polynomial $p\in\mathcal P(\Gamma(\X))$ has a representation \eqref{vcytswu6453e6} through the falling factorials. In fact, \eqref{vcyds6ue4} and  \cite[Proposition~4.6]{Umbral} imply that, in formula \eqref{vcytswu6453e6}, 
\begin{gather*}
f^{(0)}=p(0),\\
f^{(n)}(x_1,\dots,x_n)=\frac1{n!}\,\big(D_{x_1}^+\dotsm D_{x_n}^+\,p\big)(0),\quad n\ge1.
\end{gather*}
Hence, by analogy with the one-dimensional case, formula \eqref{vcytswu6453e6} can be thought of as the (finite) Newton series of the polynomial $p$.

In  classical finite difference calculus,   study of Newton series is an important and 
highly non-trivial part of the theory, see e.g.\ \cite{Gelfond}. In the infinite dimensional setting, 
spaces of functions on $\Ga(\X)$ represented through their (generally speaking infinite) Newton series appear to be also very useful.   In this section, we discuss two constructions of   such spaces. 

\subsection{Banach spaces}\label{ftayfkt}

Let $q>0$. For $n\in\mathbb N$, we denote by $\mathcal F_n(L^1(\R^d,q\,dx))$ the subspace of 
$$\textstyle L^1((\R^d)^n,\frac{1}{n!}\,q^n\,dx_1\dotsm dx_n)$$ that consists of all functions  from the latter space that are a.e.\ symmetric. Denote also $\mathcal F_0(L^1(\R^d,q\,dx)):=\R$. We then define a Banach space $\mathcal F(L^1(\R^d,q\,dx))$ as follows: its elements are all infinite sequences $f=(f^{(n)})_{n=0}^\infty$, where $f^{(n)}\in \mathcal F_n(L^1(\R^d,q\,dx))$ and
$$\|f\|_{\mathcal F(L^1(\R^d,q\,dx))}:=\sum_{n=0}^\infty\|f^{(n)}\|_{\mathcal F_n(L^1(\R^d,q\,dx))}<\infty.$$

Each sequence $f=(f^{(n)})_{n=0}^\infty\in \mathcal F(L^1(\R^d,q\,dx))$ determines, at least formally, a Newton series (cf.~\eqref{NP})
\begin{equation}\label{vufew64}
F(\gamma)=f^{(0)}+\sum_{n=1}^\infty \frac{1}{n!}\langle (\gamma)_n,f^{(n)}\rangle,\quad \gamma\in\Gamma(\X).\end{equation}
We denote by $N_q$ the space of all such (formal) series. Furthermore, by defining a norm on $N_q$ by 
$$\|F\|_{N_q}:= \|f\|_{\mathcal F(L^1(\R^d,q\,dx))},$$
we obtain a Banach space  $N_q$.

Below, for $z>0$, we will denote by $\pi_z$ the Poisson measure $\pi_\sigma$ on $\Gamma(\X)$ with intensity measure $\sigma(dx)=z\,dx$.

\begin{proposition}\label{ctydey7i}
Let $q\ge z>0$. Then, for each $F\in N_q$ of the form \eqref{vufew64}, the series on the right hand side of \eqref{vufew64} converges absolutely $\pi_z$-a.e.\ on $\Gamma(\X)$. Furthermore, $N_q\subset L^1(\Gamma(\X),\pi_z)$ and, for each $F\in N_q$, 
$$\|F\|_{L^1(\Gamma(\X),\pi_z)}\le \|F\|_{N_q} .$$
\end{proposition}

\begin{proof} By Lemma~\ref{FFN} (see also formula \eqref{FgftyF} from its proof), we have:
\begin{align}
\int_{\Gamma(\X)}\sum_{n=1}^\infty \frac{1}{n!} \langle(\ga)_n,|f^{(n)}|\,\rangle\,\pi_z(d\gamma)&\le \sum_{n=1}^\infty\frac{1}{n!} \int_{\Gamma(\X)} \langle(\ga)_n,|f^{(n)}|\,\rangle\,\pi_z(d\gamma)\notag\\
&= \sum_{n=1}^\infty \frac{z^n}{n!}\int_{(\X)^n}|f^{(n)}(x_1,\dots,x_n)|\, dx_1\dotsm dx_n\label{vcys5w}\\
&\le \sum_{n=1}^\infty \frac{q^n}{n!}\int_{(\X)^n}|f^{(n)}(x_1,\dots,x_n)|\, dx_1\dotsm dx_n<\infty.\notag
\end{align}
From here the statement easily follows.
\end{proof}

\begin{remark}
Note that, if in the series \eqref{vufew64}, $f^{(0)}\ge0$ and all functions $f^{(n)}(x_1,\dots,x_n)\ge0$ for $dx_1\dotsm dx_n$-a.a.\ $(x_1,\dots,x_n)\in(\X)^n$, then formula \eqref{vcys5w} implies that
$$\|F\|_{L^1(\Gamma(\X),\pi_z)}= \|F\|_{N_z} .$$
\end{remark}

\begin{remark} A weak point of Proposition~\ref{ctydey7i} is that it does not specify explicitly for which $\gamma\in\Gamma(\X)$ the Newton series \eqref{vufew64} converges for a given $F\in N_q$. 
\end{remark}

\begin{remark} One may use the $N_q$ spaces as appropriate  functional spaces for the study of finite difference Markov generators
and corresponding Markov semigroups. We expect that an explicit form of action of a finite difference generator  on falling factorials should give us a control of such a generator and a possibility to apply semigroup theory. 
We hope to develop such a theory in a  forthcoming paper.
\end{remark}

\subsection{Nuclear spaces}
In this section, it will be convenient for us to deal with complex-valued functions. More precisely, we will consider Newton series of the form \eqref{vufew64} with symmetric functions $f^{(n)}:(\X)^n\to\mathbb C$ being complex-valued. Hence, the function $F(\gamma)$ defined through \eqref{vufew64} will also be complex-valued. 

For each $p\in\mathbb N$, we denote by $S_p(\R^d)$ the (complex) Sobolev space $W^{p,2}(\R^d,(1+|x|^2)^pdx)$, i.e., the closure of $C_0(\R^d;\mathbb  C)$ in the Hilbertian norm
$$\|f\|_p=\bigg(\int_{\X}\sum_{|\alpha|\le p}|D^\alpha f(x)|^2(1+|x|^2)^pdx\bigg)^{1/2}.$$
Then the projective limit
$$S(\X)=\operatornamewithlimits{proj\,lim}_{p\to\infty}S_p(\X)$$
is the Schwartz space of rapidly decreasing smooth functions on $\X$, see e.g.\ \cite[Chapter 14, Section~4.3]{BSU}. Furthermore, for each $n\ge 1$, we define the $n$th symmetric tensor power of $S(\X)$ by 
$$S(\X)^{\odot n}=\operatornamewithlimits{proj\,lim}_{p\to\infty}S_p(\X)^{\odot n},$$
where the Hilbert space $S_p(\X)^{\odot n}$ is the $n$th symmetric tensor power of $S_p(\R^n)$. As easily seen, $S(\X)^{\odot n}$ is the subspace of $S(\R^{dn})$ consisting of all symmetric functions $f^{(n)}:(\X)^n\to\mathbb C$ from $S(\R^{dn})$.

Let $S'(\R^d)$ denote the dual space of $S(\X)$, i.e., the Schwartz space of tempered distributions. We have
$$S'(\X)=\operatornamewithlimits{ind\,lim}_{p\to\infty}S_{-p}(\X),$$
where $S_{-p}(\X)$ is the dual space of $S_p(\X)$ with respect to the center space $L^2(\X,dx)$. We similarly have
$$S'(\X)^{\odot n}=\operatornamewithlimits{ind\,lim}_{p\to\infty}S_{-p}(\X)^{\odot n}.$$

For each $\omega\in S'(\X)$, we define $(\omega)_n\in S'(\X)^{\odot n}$ by formula \eqref{cxrw435q}.

Let $\beta\in(0,1]$. Similarly to Subsection~\ref{ftayfkt},  we denote by $\Phi_\beta(S'(\X))$ the space of formal Newton series
\begin{equation}\label{ydfwy64}
F(\omega)=f^{(0)}+\sum_{n=1}^\infty \frac{1}{n!}\langle (\omega)_n,f^{(n)}\rangle,\quad \omega\in S'(\X)\end{equation}
such that, for each $n\in\mathbb N$, we have  $f^{(n)}\in S(\X)^{\odot n}$ and for any $p,q\in\mathbb N$,
\begin{equation}\label{ftys5wqa4}
\|F\|_{p,q,\beta}:=|f^{(0)}|+\sum_{n=1}^\infty (n!)^{1/\beta}q^n\|f^{(n)}\|_p<\infty. \end{equation}
The norms in \eqref{ftys5wqa4} determine a nuclear space topology in $\Phi_\beta(S'(\X))$, see \cite[Remark~2.11]{Sheffer}.

\begin{remark}Let $f^{(n)}\in S(\X)^{\odot n}$. Define $\Ga_0(\X): =\bigsqcup_{n=0}^\infty \Ga^{(n)} (\X)$, where $\Ga^{(n)}(\X)$ is the set of $n$-point subsets (configurations) of $\X$. 
    As well known, $\Gamma_0(\R^d)\subset S'(\X)$, hence for each $\gamma\in \Gamma_0(\R^d)$, we have $(\gamma)_n\in S'(\X)^{\odot n}$, and so the dual pairing $\langle (\gamma)_n,f^{(n)}\rangle$ is well defined. We note that the function $\Gamma_0(\X)\ni\gamma\mapsto\langle (\gamma)_n,f^{(n)}\rangle$ uniquely identifies the function $f^{(n)}$. Hence,  the falling factorial $\Gamma(\X)\cap S'(\X)\ni\gamma\mapsto\langle (\gamma)_n,f^{(n)}\rangle$ admits a unique extension to a falling factorial $S'(\X)\ni\omega\mapsto \langle(\omega)_n,g^{(n)}\rangle$ with $g^{(n)}\in S(\X)^{\odot n}$, i.e., we must have the equality $g^{(n)}=f^{(n)}$. 
\end{remark}

We denote by $\mathcal E^\beta_{\text{min}}(S'(\X))$ the space of entire functions on $S'(\X)$ of order at most $\beta$ and minimal type (when the order is equal to $\beta$.) In other words, $\mathcal E^\beta_{\text{min}}(S'(\X))$ consists of all entire functions $F:S'(\X)\to\mathbb C$ that satisfiy
\begin{equation}\label{cxrw65u432w}
\mathbf n_{p,q,n,\beta}(F):=\sup_{\omega\in S_{-p}(\X)}|F(\omega)|\exp(-q^{-n}\|\omega\|^{\beta}_{S_{-p}(\X)})<\infty. \end{equation}
See e.g.\ \cite[Section~2]{Sheffer} and the references therein for details.
The norms in \eqref{cxrw65u432w} determine a nuclear space topology in  $\mathcal E^\beta_{\text{min}}(S'(\X))$, see \cite[Remark~2.11]{Sheffer}.

The following result is a consequence of \cite[Proposition~4.1]{Sheffer}.

\begin{proposition}
We have the following equality of topological vector spaces:
$$\Phi_\beta(S'(\X))= \mathcal E^\beta_{\mathrm{min}}(S'(\X)).$$
Furthermore, for each function $F\in\mathcal E^\beta_{\mathrm{min}}(S'(\X))$ its Newton series converges in the topology of the space 
 $\mathcal E^\beta_{\mathrm{min}}(S'(\X))$.
\end{proposition}

\subsection*{Acknowledgments}

MJO was supported by the Portuguese national funds through FCT--Funda{\c c}\~ao para a Ci\^encia e a Tecnologia, within the project UIDB/04561/2020 (\url{https://doi.org/10.54499/UIDB/04561/2020}).

\end{document}